\documentclass[11pt,a4paper]{article}
\usepackage[utf8]{inputenc}
\usepackage[T1]{fontenc}
\usepackage{amssymb,amsmath, amsthm, amsfonts,url,pgfplots,hyperref,color,graphicx,cleveref,framed,hyperref,geometry,amssymb
,mathtools,nicefrac}
\usepackage{listings}

\newtheorem{theorem}{Theorem}
\newtheorem{lemma}{Lemma}
\usepackage{xcolor}
\usepackage{xcolor}

\definecolor{codegreen}{rgb}{0,0.6,0}
\definecolor{codegray}{rgb}{0.5,0.5,0.5}
\definecolor{codepurple}{rgb}{0.58,0,0.82}
\definecolor{backcolour}{rgb}{0.95,0.95,0.92}

\lstdefinestyle{mystyle}{
    backgroundcolor=\color{backcolour},   
    commentstyle=\color{codegreen},
    keywordstyle=\color{magenta},
    numberstyle=\tiny\color{codegray},
    stringstyle=\color{codepurple},
    basicstyle=\ttfamily\footnotesize,
    breakatwhitespace=false,         
    breaklines=true,                 
    captionpos=b,                    
    keepspaces=true,                 
    numbers=left,                    
    numbersep=5pt,                  
    showspaces=false,                
    showstringspaces=false,
    showtabs=false,                  
    tabsize=2
}
\def\myauthor{Benjamin Heymann\\
Criteo AI Lab, Fairplay joint team, 
Paris, France}
\def\mytitle{Side-by-side first-price auctions with imperfect bidders}
\def\myabstract{We model a procurement scenario in which two \textit{imperfect} bidders act simultaneously on behalf of a single buyer, a configuration common in display advertising and referred to as  \textit{side-by-side bidding}  but largely unexplored in theory. We  prove that the iterated best response algorithm converges to an equilibrium under standard distributional assumptions and provide sufficient condition for uniqueness. Beyond establishing existence and convergence, our analysis provides a tractable numerical method for  quantitative studies of side-by-side procurement.
} \title{\mytitle}
\author{\myauthor}

\begin{document}
\maketitle
\begin{abstract}
\myabstract
\end{abstract}
\section{Context}

A buyer hires two bidding agents to purchase inventory through first-price auctions. The rationale behind using two agents is that they may each have different strengths, such as access to unique suppliers. However, when both agents participate in the same auction, they compete against each other, which has a direct and an indirect effect.

The direct effect is mechanical: when one agent bids higher, it reduces the other’s probability of winning, and the buyer ultimately pays the higher of the two bids. The indirect effect arises through learning, as each agent infers the other’s bidding strategy from observed auction outcomes and gradually adapts its own bids. This indirect effect is harder to analyze without reference to an equilibrium concept, since it depends on the mutual adaptation of both agents’ strategies.

Having two  bids for the same buyer is called side-by-side bidding, and  is a well known problem in  ad-tech  that already existed before Google switch to the first price rule. At that time, the auctions mostly followed a second pricing scheme, which made the effect even worse~\cite{heymann2020bid} for the buyer.

 In this note, we show that sequential best-response iterations  converge to a unique equilibrium in the specific case of two bidders acting on behalf of a single buyer. 
Best-response iterations is a very old adaptive procedure  that was first  described for Cournot's oligopoly model~\cite{cournot1863principes} even  before the  formalization of game theory by von Neumann and Morgenstern~\cite{von1947theory}.
 In general, computing equilibria in first-price auctions requires specialized methods, as standard learning dynamics --- such as best-response, but also fictitious play and no-regret updates--- do not guarantee convergence to an equilibrium.
 
 Our analysis assumes that competition from other market participants is exogenous and fixed, reflecting a standard  assumption in large markets with small actors. An important modeling and technical assumption is that we consider \textit{imperfect bidders} -- bidders that bid a noisy version of their best response to the competition, which can either be  due to execution inaccuracies or deliberate exploration for learning purposes.  We further make the simplifying assumption that  the multiplicative noise distribution is stationary over time.

The concept of imperfect bidders is particularly pertinent  to analyze  the winner's curse, whereby the winning bidder pays more than necessary due to incomplete information. In side-by-side settings, imperfect bidders may exacerbate this issue by making noise-influenced bids that can cause overbidding, leading to higher costs for the buyer.

In terms of technical tools, we rely on the notion of log-concavity, which is a classic assumption in economics~\cite{bagnoli2005log}. 
A function $f:\mathbb{R^+}\to\mathbb{R^+_*}$ is said to be \textit{log-concave} when $\log f$ is concave. The convergence result relies on Berge's maximum principle~\cite{berge1877topological}, and the uniqueness result on Tarki fixed-point theorem~\cite{tarski1955lattice}, as it allows us to extract minimal and maximal putative equilibria.

\section{Model}

We use the index $i \in \{1,2\}$ to refer to a given bidder, and $-i$ to refer to the other. 
The bidders use the same value, $v > 0$, which is 
 determined by the buyer.
The submitted bid $\hat{b}_i$ is a perturbed version of the payoff maximizing bid $b^\star_i$, either due to execution inaccuracies or deliberate exploration for learning purposes, formally
\begin{align}
\label{eq:bhat}
\hat{b}_i &= b_i^\star  \cdot \varepsilon_i \ , \\
\pi_i(b,b_{-i}^\star)&= (v - b) \cdot Q(b) \cdot F_{-i}(b)\ ,\\
\label{eq:bstar}
    b_i^\star \in BR_i(F_{-i}) &= \arg\max_b \pi_i(b,b_{-i}^\star)\ ,
\end{align}
where 
$Q$
    is the CDF of the highest bid of the exogenous competition. 
    It is supposed smooth, log-concave and with  $Q'(t)>0$ for $t\in [0,v]$;
   $\varepsilon_i$ is a random variable of  log-concave  CDF $N_i$ on $\mathbb{R}_+$\footnote{    Typically, we might have a multiplicative    lognormal noise of mean $1$, but our results hold with a  larger class of assumption. Also, those  results could be adapted to additive noise}, and $F_{-i} = N_{-i}(\nicefrac{\cdot} {b_{-i}})$  is the CDF of $\hat{b}_{-i}$.
Since  $F_i$ only depends on $b_i$,  we can abusively write
$BR_i(F_{-i}) = BR_i(b_{-i}^\star)$.

In first-price auction, the study of side-by-side bidding requires accounting for the game dynamics, otherwise, the model  predicts  an unrealistic outcome. The noise plays a key role here. First, it allows justifying  why  the bidders have access to payoff estimates. 
Indeed, the perturbation of the best response  provides a way to model how bidders anticipate what would happen under alternative actions, and recover a realistic notion of equilibrium.
 Second, it smooths the dynamics in a way that is favorable for analysis. 
Third, the perturbation also allows us to account for bidder heterogeneity by encoding differences in performance within the parameter~$\varepsilon_i$.

\section{Convergence of best   response iterations}

Our first result is that the \textbf{best-response iterations converge to an equilibrium}.
\begin{theorem}
\label{th:cv}
Let 
$b_2^{0}$ from $[0,v]$.
Let
    \begin{align*}
    b_1^{(k)}&\in BR_1(b_{2}^{(k-1)}),\quad &\quad & \forall k\in\mathbb{N}^\star,\\
        b_2^{(k)}&\in BR_2(b_{1}^{(k)}),\quad &\quad & \forall k\in\mathbb{N}^\star,
    \end{align*}
    then 
    \begin{align}
     &\exists \lim_{k\to \infty}  b_i^{(k)} = b_i^\star \quad &\forall i\in\{1,2\} \tag{limit}\\
     &b_i^{\star}\in BR_i(b_{-i}^{\star}) \quad &\forall i\in\{1,2\} \tag{Equilibrium}.
    \end{align}
\end{theorem}
\begin{proof}
We first show that the best response $BR^i$ is  increasing in $b^{-i}$.
Since $Q$ and $F_{-i}$ are differentiable, 
 so is  their product
 $\gamma(b):=  Q(b) \cdot F_{-i}(b)$. 
Therefore by~\eqref{eq:bstar} the optimal bid
$b_i^\star$ satisfies the first order condition
\begin{align}
\label{eq:foc}
    \frac{d}{db_i^\star}[(v-b_i^\star) \cdot \gamma(b_i^\star)] = 0.
\end{align}
By definition of $Q$, $\gamma'>0$ on $[0,v]$, therefore~\eqref{eq:foc}  simplifies into
\begin{align*}
b_i^\star=v-\nicefrac{\gamma(b_i^\star)}{\gamma'(b_i^\star)}= \phi(b_i^\star).  
\end{align*}
Since $Q$ and $F_{-i}$ are by definition log-concave, so is there product $\gamma$, hence by Lemma~\ref{lemma:fp}, $b^\star_i$ is uniquely defined.

Another equivalent way to characterize $b^\star_i$ is to apply to apply a first order condition on the logarithm of the expected payoff $\pi_i$, in which case, we get the condition
\begin{align}
    \underbrace{\frac{1}{v-b_i^\star}}_{\text{increasing in }b_i^\star} = \underbrace{\frac{Q'(b_i^\star)}{Q(b_i^\star)}}_{\text{decreasing in }b_i^\star}+ \underbrace{\frac{F'_{-i}(b_i^\star)}{F_{-i}(b_i^\star)}}_{\text{decreasing in }b_i^\star \text{, increasing in }b_{-i}^\star}
\end{align}
From there, it is easy to show ad absurdum that the best response $BR_i$ is non-decreasing in $b_{-i}^\star$
.

If $b_2^{(1)}\geq b_2^{(0)}$, then by monotony of $BR_1$, 
$b_1^{(2)}\geq b_1^{(1)}$, hence 
by monotony of $BR_2$,
$b_2^{(2)}\geq b_2^{(1)}$, hence the sequences $b_i^{(k)}$, for $i\in\{1,2\}$ are non-decreasing.
Similarly, if  $b_2^{(1)}\leq b_2^{(0)}$, then by monotony of $BR_1$, 
$b_1^{(2)}\leq b_1^{(1)}$, hence 
by monotony of $BR_2$,
$b_2^{(2)}\leq b_2^{(1)}$, hence the sequences $b_i^{(k)}$, for $i\in\{1,2\}$ are non-increasing.
Overall, the sequences $b_i^{(k)}$, for $i\in\{1,2\}$ are monotone.
Since they are valued in $[0,v]$, they converge to limits that we denote by $b^\star_i$.
Then, because  $\pi_i$ is continuous in $b$ and $b_{-i}$, Berge's Maximum principle~\cite{berge1877topological} implies that the argmax is upper-hemicontinuous~\footnote{A correspondence \( F: X \rightrightarrows Y \) is \textbf{upper hemicontinuous} if, whenever \( x_n \to x \) and \( y_n \in F(x_n) \) with \( y_n \to y \), we have $y\in F(x)$
}, it follows that $b_i^{\star}\in BR_i(b_{-i}^{\star})$.
\end{proof}

Our next result states that \textbf{the equilibrium of the best-response dynamics defined in Theorem~\ref{th:cv} is unique}. This uniqueness is not trivial. Indeed, we provide in the appendix an example violating the assumption on the support of $\varepsilon$ for which the uniqueness is not satisfied.

\begin{theorem}
The equilibrium $(b_1^\star,b_2^\star)$ does not depend on the starting point $b_2^{0}$.
\end{theorem}

We leverage Tarski's theorem, which applies to monotone operators on complete lattices, which tells us that the set of fixed points is a  non-empty complete lattice. 

\begin{proof}
By Tarski~\cite{tarski1955lattice,topkis}, the set of equilibrium in a non-empty complete lattice. 
So it make sense to denote by $(b_1^-,b_2^-)$ and
        $(b_1^+,b_2^+)$  the biggest and smallest equilibrium. 
        Without loss of generality, we suppose, 
        $\nicefrac{b_1^+}{b_1^-}\geq \nicefrac{b_2^+}{b_2^-}\geq 1$, and we set $\alpha = \nicefrac{b_1^+}{b_1^-}$.
 We  then show the following relations: 
\begin{align}
\label{eq:trick}
    b_1^+ \underbrace{\leq}_{\clubsuit} BR_1(\alpha b_2^-)\underbrace{<}_{\spadesuit} \alpha b_1^-
\end{align}
We start with $\clubsuit$. By assumption
    $\alpha b_2^-\geq b_2^+$ and $BR_1$ is non-decreasing (Proof of Theorem~\ref{th:cv}), therefore 
        $b_1^+ = BR_1(b_2^+)\leq BR_1(\alpha b_2^-)$.
 
Now, to show  $\spadesuit$, by definition
        \begin{align*}
        BR_1(\alpha b_2^-) &= \arg\max_b (v - b) \cdot Q(b) \cdot \Pr(\varepsilon_2\alpha b_2^- \leq b)\\
    \intertext{by change of variable $b=\alpha b'$}
        &=
        \alpha\arg\max_{b'} (v - \alpha b') \cdot Q(\alpha b') \cdot \Pr(\varepsilon_2\alpha b_2^- \leq \alpha b')\\
     \intertext{\text{Lemma}~\ref{lemma:homogeneity}}
       &<\alpha\arg\max_{b'} (v -  b') \cdot Q( b') \cdot \Pr( \varepsilon_2 b_2^- \leq  b')\\
       & = \alpha BR_1(b_2^-)\\
       & =\alpha b_1^-.
    \end{align*}

 We now have proven~\eqref{eq:trick}, which is in contradiction with the definition of $\alpha$.
We conclude that  the equilibrium is unique. 
\end{proof}

The convergence of the best-response iteration provides a natural justification for defining its limit as an equilibrium concept.
On top of that,  it provides a theoretical validation for estimating market outcomes through  the best-response iterations.

\section{Technical lemmata}

\begin{lemma}
\label{lemma:fp}
Let $\gamma:[0,v]\to [0,\infty[$ be a continuously differentiable function such that 
$\log\gamma$ is concave on $[0,v]$ and $\gamma'>0$.  
Then function
$
  \Phi(b) = v - \frac{\gamma(b)}{\gamma'(b)}
$
 is monotone on $[0,v]$, in particular, any fixed point of $\Phi$ is unique.
\end{lemma}

\begin{proof}
Concavity of $\log\gamma$ implies that 
$
 \frac{\gamma'(x)}{\gamma(x)}
$
is nonincreasing on $(0,v)$,  therefore its inverse,  $\frac{\gamma(x)}{\gamma'(x)}$ is nondecreasing, and therefore $\Phi(b)$ is nonincreasing.
Therefore the solution to $b=\Phi(b)$, if it exists, is unique.
\end{proof}

\begin{lemma}
\label{lemma:homogeneity}
For any positive numbers $b$ and $b^\star_{-i}$,
           $\arg\max_{b} \pi(\alpha b,\alpha b_{-i}^\star)$ is strictly decreasing in $\alpha$.
\end{lemma}

\begin{proof}
   \begin{align*}
   b\in \arg\max_{b} \pi(\alpha b,\alpha b_{-i}^\star)\\
   &\iff    b\in  \arg\max_{b'} (v - \alpha b') \cdot Q(\alpha b') \cdot \Pr(\varepsilon_{-i}\alpha b_{-i}^\star\leq \alpha b')\\
        &\implies  \frac{\alpha}{v - \alpha b} = \alpha\frac{Q'(\alpha b)}{Q(\alpha b)}+ \frac{F_i'(b)}{F_i(b)} \\
             &\implies  \frac{1}{v - \alpha b} = \frac{Q'(\alpha b)}{Q(\alpha b)}+ \frac{ F_i'(b)}{\alpha F_i(b)} \ .
   \end{align*} 
 The LHS is strictly increasing in $\alpha$ and $b$, while the RHS is non-increasing in $\alpha$ and $b$. The conclusion follows.
\end{proof}
We next illustrate why  the equilibrium $b^\star$ might not be unique if the assumptions are not satisfied. 
\begin{lemma}[Non-uniqueness in general]
   Take $v=1$ and  $N_1=N_2=Q = \min(x,1)$, then any value between $1/2$ and $2/3$ is a valid value for $b^\star$. 
\end{lemma}
\begin{proof}
The profit for bidder $i$ writes
\begin{align*}
    \pi_i(b_i,b_{-i}) = (v - b_i) \min(b_i,1) \cdot \min(\nicefrac{b_i}{b_{-i}},1) =\begin{cases}
     (v - b_i) b_i^2/b_{-i} \quad & \text{if} \quad b_i\leq b_{-i}\leq 1 \\
      (v - b_i) b_i \quad &\text{if} \quad 1\geq b_i\geq b_{-i} \ .
    \end{cases}
\end{align*}
Therefore if $b_{-i}\in ]1/2,2/3[$
\begin{align*}
    \frac{\partial\pi_i(b_i,b_{-i})}{\partial b_i} =  \begin{cases}
     \frac{b_i}{b_{-i}}(2v - 3 b_i) \quad & \text{if} \quad b_i\leq b_{-i}\leq 1 \\
      v - 2b_i  \quad &\text{if} \quad 1\geq b_i\geq b_{-i} \ .
    \end{cases}
\end{align*}
Hence  $\pi_i(b_i,b_{-i})$ is strictly increasing for $b_i<b_{-i}$ and 
strictly decreasing for $b_i>b_{-i}$.
Hence the best response to $b_{-i}$ is $b_{-i}$. Therefore $(b_{-i},b_{-i})$ is an equilibrium,  for any $b_{-i}\in ]1/2,2/3[$.
\end{proof}

 \end{document}